\documentclass[10pt,a4paper,reqno]{amsart}
\usepackage{geometry}
\usepackage{CJK}
\usepackage{amsmath}
\usepackage{epstopdf}
\usepackage{mathrsfs}
\usepackage{bbding}
\usepackage{cases}
\usepackage{subfigure}
\usepackage{float}
\usepackage{lineno,hyperref}
\usepackage{booktabs}
\usepackage{tabularx}%
\usepackage{acronym}
\usepackage{amsfonts}
\usepackage{amsmath}
\usepackage{amssymb}
\usepackage{amsthm}
\usepackage{bm}
\usepackage{braket}
\usepackage{cite}
\usepackage{color}
\usepackage{geometry}
\usepackage{graphicx}
\usepackage{hyperref}
\usepackage{mathrsfs}
\usepackage{mathtools}
\usepackage{setspace}
\usepackage{stmaryrd}
\usepackage{tikz}
\usepackage{epsfig}
\usepackage{setspace}
\usepackage{epstopdf}
\usepackage{booktabs}
\usepackage{threeparttable}
\usepackage{diagbox}

\newgeometry{top=2cm,bottom=2cm,outer=1.5cm,inner=1.5cm}
\newcommand{\bse}{\begin{subequations}}
\newcommand{\ese}{\end{subequations}}

\newtheorem{theorem}{Theorem}

\numberwithin{equation}{section}
\DeclareSymbolFont{largesymbols}{OMX}{yhex}{m}{n}
\DeclareMathAccent{\Widehat}{\mathord}{largesymbols}{"62}

\title[Breathers and rogue waves on the double-periodic background for the reverse-space-time derivative nonlinear Schr\"{o}dinger equation]{Breathers and rogue waves on the double-periodic background for the reverse-space-time derivative nonlinear Schr\"{o}dinger equation}
\author{Huijuan Zhou}
\address[HZ]{School of Mathematical Sciences, Shanghai Key Laboratory of Pure Mathematics and Mathematical Practice, and Shanghai Key Laboratory of Trustworthy Computing \\
East China Normal University \\ Shanghai 200241 \\ China}
\author{Yong Chen$^*$}
\address[YC]{School of Mathematical Sciences, Shanghai Key Laboratory of Pure Mathematics and Mathematical Practice, and Shanghai Key Laboratory of Trustworthy Computing \\
East China Normal University \\ Shanghai 200241 \\ China}
\address[YC]{College of Mathematics and Systems Science \\ Shandong University of Science and Technology \\ Qingdao 266590 \\ China}
\email{ychen@sei.ecnu.edu.cn}

\begin{document}

\begin{abstract}

In the present investigation, the solutions on the periodic and double-periodic background are successfully constructed by Darboux transformation using a plane wave seed solution. Firstly, the Darboux transformation for the reverse-space-time DNLS equation is constructed. Secondly, periodic solutions, breathers, double-periodic solutions, breathers on the periodic background and double-periodic background are studied. Thirdly, the higher-order rogue waves on the periodic and double-periodic background are constructed by semi-degenerate Darboux transformations. In addition, the dynamic behaviors of the solutions are plotted to show some interesting new solution structures.


\end{abstract}
\maketitle

\section{Introduction}
 \ \ \ \
 Nonlinear evolution equations play an important role in integrable systems and due to the applications of their solutions \cite{ablowitz-prl-1973,zakharov-springer-1991}, many well-known mathematicians and physicists did some significant work \cite{kaup-jmp-1978,olver-springer-1990,2-m-Abrarov-1992,yajima-jpsj-1995,lederer-pr-2008}.
  Breathers and rogue waves are the important solutions of nonlinear evolution equations. There have been a lot of studies about breathers and rogue waves in recent years \cite{freak1996,rogue2007nature,rogue2009pra,frw2011,lcz2014,breather2014,wanglei-annals-2015,Liming-2018,jxw2020}.

The derivative nonlinear Schr\"{o}dinger equation (DNLS) \cite{2-m-Kundu-2010,JPAXSW,zys2014,xt2018nd} is given by
\begin{equation}\label{dnls}
iq_{t}-q_{x x}+i(q^{2}q^{*})_{x}=0,
\end{equation}
where the complex function $q=q(x,t)$ denotes the wave envelope and $^{*}$ denotes the complex conjugation. Eq. \eqref{dnls} arises in the study of circular polarized Alfv\'{e}n waves in plasma \cite{plasma1}, propagating parallel to the magnetic field \cite{magnetic}, which is one of the most important integrable systems in mathematics and physics. Recently, the equation is also used to describe large-amplitude magnetohydrodynamic waves \cite{plasma2,plasma3} of plasmas, nonlinear optics, the sub-picosecond and femtosecond pulses in single-mode optical fiber \cite{7,8,9}. The DNLS and nonlocal DNLS  equations are reduced from the Kaup-Newell system \cite{kn1994,kn1999-zz} and are Lax integrable. There generate many new physical phenomena and have important physical significance when nonlocal terms are added to nonlinear equations.

In recent years, many researchers have studied  nonlocal DNLS equations from different viewpoints and perspectives. For example, in \cite{CNSCZZX}, the global bounded solutions of the nonlocal DNLS equation have been obtained from zero seed solution by Darboux transformation (DT) \cite{29i,VBDB,CHDB,2-m-GuChaohao-2005,xtzs2017,wmmnd}. Furthermore, solutions and connections of nonlocal DNLS equations have been studied in \cite{NDSY}. In \cite{MPLBMDX}, the periodic bounded solutions of the second-type nonlocal DNLS equation from zero seed solutions have been studied. The $PT$-symmetric, reverse-time, and reverse-space-time nonlocal DNLS equations are integrable infinite dimensional Hamiltonian dynamical systems, which were first proposed by Ablowitz and Musslimani \cite{ab-prl-2013,ab-studies-2017}. The general N-solitons in these three nonlocal nonlinear Schr\"{o}dinger equations are presented by Yang in \cite{yjk-pla-2019}. To investigate the connections between solutions at reverse-space-time points $(x, t)$ and $(-x, -t)$, we need to consider the reverse-space-time reduction. The reverse-space-time DNLS equation is as follows
\begin{equation}\label{ndnls}
 iq_{t}-q_{x x}+i(q^{2}q(-x,-t))_{x}=0,
\end{equation}
here the symmetry reductions are nonlocal both in space and time. The reverse-space-time DNLS equation has many physical applications in optics, ocean water waves, quantum entanglement and an unconventional system of magnetics etc \cite{gad-pra-2016,yjk-pre-2018,pjc-nd-2021}. For eq. \eqref{ndnls}, the evolution of the solution at location $(x, t)$, depends both on the local position $(x, t)$ and the distant position $(-x, -t)$. This implies that the states of the solution at distant opposite locations are directly related, reminiscent of quantum entanglement in pairs of particles \cite{yjk-pla-2019}. The solution of reverse-space-time DNLS equation can extend the solution of the local equation to a more general case and deepen the physical research on the mechanism of ocean rogue waves. These results would also be useful for understanding the corresponding rational soliton phenomena in many fields of nonlocal nonlinear dynamical systems such as nonlinear optics, Bose-Einstein condensates, ocean and other relevant fields \cite{zgq2017,ablowitz-tmp-2018}.

In general, it is extremely nontrivial to construct the rogue waves on a periodic background which is usually associated with complicated  Jacobi elliptic functions \cite{kaa-2014,cjb2018non,rjg2018,cjb-2019,xbgxg2020,lb2020,zhq2021,sms-2021}, PT symmetry \cite{hx2016epjp}, integrable equations with variable coefficients \cite{hjs2013nmp,hjs2014pla}, or vector integrable equations \cite{zlc2017cnsns}. In \cite{lwhjs2018,dccgyt2019}, the rogue waves on the periodic background have been constructed by using odd-order semi-degenerate DT. In this article, we mainly study the breathers and the rogue waves on the double periodic background by using even-fold DT and even-order semi-degenerate DT respectively. This is an effective new method to construct the solutions on double-periodic background without using Jacobi elliptic functions. It is of great physical significance to study rogue waves on a double-periodic background. For example, the rogue waves on the double-periodic background in the hydrodynamical experiments are possible due to the rogue waves on the continuous wave background observed in laser optics and water tanks \cite{sr2018pre}. Rogue waves on the double-periodic background could be relevant to diagnostics of rogue waves on the ocean surface and understanding the formation of random waves due to modulation instability \cite{ac-wavemotion-2017}.

In this work, we construct the breathers and rogue waves on the periodic background by odd-fold DT and odd-order semi-degenerate DT respectively. This is the first time to extend this method to reverse-space-time nonlocal equations. Remarkably, using a plane wave seed solution, the breathers and rogue waves on the double-periodic background are first successfully constructed by even-fold DT and even-order semi-degenerate DT respectively. By taking the dynamics analysis of the first-order rogue waves on double-periodic background, we show two types of structures: the two peaks and four peaks.  The interesting thing can also be seen from the dynamic figures of the second-order rogue waves on the double-periodic background. There are two types of structures: one peak and two peaks. We shall also show the transformation process of double-periodic background and plane wave background in this study. These results have not been reported to our best knowledge.

The organizational structure of this paper is as follows. In section 2, the determinant representation of the $n$-fold DT formula is given. In section 3, using a plane wave seed solution, the periodic solution, breathers on the periodic background are given by odd-fold DT.  The double-periodic solution, breathers and breathers on the double-periodic background solution are given by even-fold DT. In section 4, we construct higher-order rogue waves on the periodic background and double-periodic background by semi-degenerate DT formula. The final section is devoted to conclusion.

\section{DT of the reverse-space-time DNLS equation}

{\bf \subsection{Lax pair of the reverse-space-time DNLS equation}}
\ \ \ \
Starting from the Kaup-Newell system \cite{FEPA,MWJPA}, when the reduction condition is $r(x, t)=-q(-x, -t)$, the Lax pair of the reverse-space-time DNLS equation can be obtained as follows.
\begin{equation}
\begin{array}{c}\label{xlax}
\Psi_{x}=\left(i\sigma \lambda^{2}+Q \lambda\right) \Psi=U \Psi,\\
\end{array}
\end{equation}
\begin{equation}
\begin{array}{c}\label{tlax}
\Psi_{t}=\left(2 i\sigma \lambda^{4}+V_{3} \lambda^{3}+V_{2} \lambda^{2}+V_{1} \lambda\right) \Psi=V \Psi,
\end{array}
\end{equation}
with
$$
\Psi=\left(\begin{array}{c}
\phi \\
\varphi
\end{array}\right), \quad \sigma=\left(\begin{array}{cc}
1 &  0 \\
0 & -1
\end{array}\right), \quad Q=\left(\begin{array}{cc}
0 & \ q \\
-q(-x,-t) & \ 0
\end{array}\right),
$$
$$
V_{3}=2Q, \quad V_{2}=iQ^{2}, \quad V_{1}=Q^{3}+iQ_{x}\sigma=\left(\begin{array}{cc}
0 & -i q_{x}-q^{2}q(-x,-t) \\
iq_{x}(-x,-t)+q(-x,-t)^{2}q & 0
\end{array}\right).
$$
Eq. \eqref{ndnls} can be derived from the integrable condition $U_{t}-V_{x}+[U, V]=0$ of Lax pair \eqref{xlax} and \eqref{tlax}.

Introducing
$\Psi_{j}=\left(\begin{array}{c}
\phi_{j} \\
\varphi_{j}
\end{array}\right)=\left(\begin{array}{c}
\phi_{j}\left(x, t, \lambda_{j}\right) \\
 \varphi_{j}\left(x, t, \lambda_{j}\right)
\end{array}\right),  j=1,2, \ldots, $ which is the eigenfunction of the  Lax pair \eqref{xlax} and \eqref{tlax} associated with $\lambda=\lambda_{j}$. The eigenfunction admit the following symmetry condition:
\begin{equation}\label{sc}
\left(\begin{array}{l}
\phi(x, t; \lambda_{j}) \\
\varphi(x, t; \lambda_{j})
\end{array}\right)
=
\left(\begin{array}{l}
\varphi(-x, -t; \lambda_{j}) \\
\phi(-x, -t; \lambda_{j})
\end{array}\right).
\end{equation}

{\bf  \subsection{ $n$-fold DT of reverse-space-time DNLS equation} }
\ \
DT has unique advantages in constructing solutions due to pure algebraic construction. In this section, the DT for the Eq. \eqref{ndnls} will be introduced. Under gauge transformation
\begin{equation}\label{gt}
\Psi^{[1]}=T \Psi,
\end{equation}  the spectral problem \eqref{xlax} and \eqref{tlax} can be transformed to
\begin{equation}\label{nlax}
\begin{array}{l}
\Psi^{[1]}_{x}=(i\sigma\lambda^{2}+Q^{[1]}\lambda) \Psi^{[1]}=U^{[1]} \Psi^{[1]},  \\
\Psi^{[1]}_{t}=(2 i\sigma \lambda^{4}+V^{[1]}_{3} \lambda^{3}+V^{[1]}_{2} \lambda^{2}+V^{[1]}_{1} \lambda) \Psi^{[1]}=V^{[1]} \Psi^{[1]}.
\end{array}
\end{equation}
Where
$$
Q^{[1]}=\left(\begin{array}{cc}
0 & \ q[1] \\
-q[1](-x,-t) & \ 0
\end{array}\right),
$$
$$
V^{[1]}_{3}=2Q^{[1]}, \quad V_{2}^{[1]}=iQ^{[1]2},
$$
$$
V_{1}^{[1]}=Q^{[1]3}+iQ^{[1]}_{x}\sigma=\left(\begin{array}{cc}
0 & -i q[1]_{x}-q[1]^{2}q[1](-x,-t) \\
iq[1]_{x}(-x,-t)+q[1]^{2}(-x,-t)q[1] & 0
\end{array}\right).
$$
After derivation, we get the following conclusion.
\begin{equation}\label{3.3}
T_{x}=U^{[1]}T-TU,
\end{equation}
\begin{equation}\label{3.4}
T_{t}=V^{[1]}T-TV.
\end{equation}
Furthermore, the following identity can be deduced
\begin{equation}\label{4.8}
U_{t}^{[1]}-V^{[1]}_{x}+[U^{[1]},V^{[1]}]=T(U_{t}-V_{x}+[U,V])T^{-1}.
\end{equation}
Due to the matrix $T$ is nonsingular, the  zero curvature equation $U_{t}-V_{x}+[U,V]=0$ is equivalent to $U_{t}^{[1]}-V^{[1]}_{x}+[U^{[1]},V^{[1]}]=0.$ This implies that, in order to make spectral problem Eq. \eqref{xlax} is invariant under the gauge transformation Eq. \eqref{gt}, it is important to find a matrix $T$ so that $U^{[1]}$ and $V^{[1]}$ have the same forms as $U$ and $V$. At the same time, the old solutions $(q, q(-x,-t))$ in spectral matrixes $U$ and $V$ are mapped into new solutions $(q[1], q[1](-x,-t))$ in transformed spectral matrixes $U^{[1]}$ and $V^{[1]}$.

In general, if the Darboux matrix $T$ is a polynomial of the parameter $\lambda$, for simplicity, we take $T$ as
\begin{equation}\notag
T=\left(\begin{matrix}
a_{1}&b_{1}\\
c_{1}&d_{1}\\
\end{matrix}
\right)\lambda
+\left(\begin{matrix}
a_{0}&b_{0}\\
c_{0}&d_{0}\\
\end{matrix}
\right).
\end{equation}
Substituting the Darboux matrix $T$ into Eq. \eqref{3.3} and  Eq. \eqref{3.4}, the one-fold DT formula can be derived by comparing the coefficient in terms of  $\lambda^{i}$
\begin{equation} \label{db}
q[1]=\frac{a_{1}}{d_{1}}q-2i\frac{b_{0}}{d_{1}}.
\end{equation}
We also can deduced that $b_{1}=c_{1}=0$, $a_{1}$ and $d_{1}$ are undetermined functions about $x$ and $t$. $a_{0}$, $b_{0}$, $c_{0}$ and $d_{0}$ are constants. In order to obtain the specific expression of the elements in the matrix $T$, for simplicity, let $a_{0}=d_{0}=0$ then
\begin{equation}\label{dtT}
T_{1}=\left(\begin{array}{cc}
a_{1} & 0 \\
0 & d_{1}
\end{array}\right) \lambda+\left(\begin{array}{cc}
0 & b_{0} \\
c_{0} & 0
\end{array}\right).
\end{equation}
In particular, taking $b_{0}=c_{0}=\lambda_{1}$, the one-fold DT formula is given by the eigenfunction $\Psi_{1}$ associated with $\lambda_{1}$ as follows
\begin{equation}
\begin{array}{l}\label{21}
q[1]=\left(\frac{\varphi_{1}}{\phi_{1}}\right)^{2} q+2 i \frac{\varphi_{1}}{\phi_{1}} \lambda_{1}.\\
\end{array}
\end{equation}

After $n$ times iterations based on the one-fold Darboux matrix \eqref{dtT}, the form of $n$-fold Darboux matrix is as follows
 \begin{equation}\notag
T_{n}=\sum_{i=0}^{n} P_{i} \lambda^{i},
\end{equation}
$$
P_{i}=\left\{
  \begin{array}{ll}
   \left(\begin{array}{ll}
a_{i}\quad & 0 \\
0\quad & d_{i}
\end{array}\right), & \hbox{$i=n-2\ell$, $0\leq \ell\in Z \leq \frac{n}{2}$;} \\ \\
    \left(\begin{array}{ll}
0 & b_{i} \\
c_{i} & 0
\end{array}\right), & \hbox{$i=n-2\ell-1$, $0\leq \ell\in Z \leq \frac{n}{2}$.}
  \end{array}
\right.
$$
where $P_{0}$ is a constant matrix and $P_{i} (1 \leq i \leq n)$ is a matrix function about $x$ and $t$.
Using the same derivation method as the one-fold DT formula, yields
\begin{equation}\label{qn}
q[n]=\frac{a_{n}}{d_{n}} q-2 i \frac{b_{n-1}}{d_{n}}.
\end{equation}
Furthermore, the determinant representation of $a_{n}$, $d_{n}$ and $b_{n-1}$ can be given by the kernel problem of DT matrix $T_{n}$. i.e.,
\begin{equation}\label{hwt}
\left.T_{n}\right|_{\lambda=\lambda_{k}} \Psi_{k}=\sum_{i=0}^{n} P_{i} \lambda_{k}^{i} \Psi_{k}=0, k=1,2, \cdots, n.
\end{equation}
Then the concrete expression of the new solutions $q[n]$ can be seen in the following.
\begin{theorem}  The new solutions  $q[n]$ given by the following $n$-fold DT formula of the Eq. \eqref{ndnls}.
\begin{equation}\label{dtf}
q[n]=\frac{W_{11}^{2}}{W_{21}^{2}} q+2 i \frac{W_{11} W_{12}}{W_{21}^{2}},
\end{equation}
where $W_{11}$, $W_{12}$, and $W_{21}$ have different forms depending on the parity of $n$.\\
When $n=2k$,
$$
W_{11}=\left|\begin{array}{cccccc}
\lambda_{1}^{n-1} \varphi_{1} & \lambda_{1}^{n-2} \phi_{1} & \lambda_{1}^{n-3} \varphi_{1} & \ldots & \lambda_{1} \varphi_{1} & \phi_{1} \\
\lambda_{2}^{n-1} \varphi_{2} & \lambda_{2}^{n-2} \phi_{2} & \lambda_{2}^{n-3} \varphi_{2} & \ldots & \lambda_{2} \varphi_{2} & \phi_{2} \\
\vdots & \vdots & \vdots & \vdots & \vdots & \vdots \\
\lambda_{n}^{n-1} \varphi_{n} & \lambda_{n}^{n-2} \phi_{n} & \lambda_{n}^{n-3} \varphi_{n} & \ldots & \lambda_{n} \varphi_{n} & \phi_{n}
\end{array}\right|, \\
\ \\
$$
$$
W_{12}=\left|\begin{array}{cccccc}
\lambda_{1}^{n} \phi_{1} & \lambda_{1}^{n-2} \phi_{1} & \lambda_{1}^{n-3} \varphi_{1} & \ldots & \lambda_{1} \varphi_{1} & \phi_{1} \\
\lambda_{2}^{n} \phi_{2} & \lambda_{2}^{n-2} \phi_{2} & \lambda_{2}^{n-3} \varphi_{2} & \ldots & \lambda_{2} \varphi_{2} & \phi_{2} \\
\vdots & \vdots & \vdots & \vdots & \vdots & \vdots \\
\lambda_{n}^{n} \phi_{n} & \lambda_{n}^{n-2} \phi_{n} & \lambda_{n}^{n-3} \varphi_{n} & \ldots & \lambda_{n} \varphi_{n} & \phi_{n}
\end{array}\right|, \\
\ \\
$$
$$
W_{21}=\left|\begin{array}{cccccc}
\lambda_{1}^{n-1} \phi_{1} & \lambda_{1}^{n-2} \varphi_{1} & \lambda_{1}^{n-3} \phi_{1} & \ldots & \lambda_{1} \phi_{1} & \varphi_{1} \\
\lambda_{2}^{n-1} \phi_{2} & \lambda_{2}^{n-2} \varphi_{2} & \lambda_{2}^{n-3} \phi_{2} & \ldots & \lambda_{2} \phi_{2} & \varphi_{2} \\
\vdots & \vdots & \vdots & \vdots & \vdots & \vdots \\
\lambda_{n}^{n-1} \phi_{n} & \lambda_{n}^{n-2} \varphi_{n} & \lambda_{n}^{n-3} \phi_{n} & \ldots & \lambda_{n} \phi_{n} & \varphi_{n}
\end{array}\right|. \\
\ \\
$$
When $n=2k+1,$
$$
W_{11}=\left|\begin{array}{cccccc}
\lambda_{1}^{n-1} \varphi_{1} & \lambda_{1}^{n-2} \phi_{1} & \lambda_{1}^{n-3} \varphi_{1} & \ldots & \lambda_{1} \phi_{1} & \varphi_{1} \\
\lambda_{2}^{n-1} \varphi_{2} & \lambda_{2}^{n-2} \phi_{2} & \lambda_{2}^{n-3} \varphi_{2} & \ldots & \lambda_{2} \phi_{2} & \varphi_{2} \\
\vdots & \vdots & \vdots & \vdots & \vdots & \vdots \\
\lambda_{n}^{n-1} \varphi_{n} & \lambda_{n}^{n-2} \phi_{n} & \lambda_{n}^{n-3} \varphi_{n} & \ldots & \lambda_{n} \phi_{n} & \varphi_{n}
\end{array}\right|, \\
\ \\
$$
$$
W_{12}=\left|\begin{array}{cccccc}
\lambda_{1}^{n} \phi_{1} & \lambda_{1}^{n-2} \phi_{1} & \lambda_{1}^{n-3} \varphi_{1} & \ldots & \lambda_{1} \phi_{1} & \varphi_{1} \\
\lambda_{2}^{n} \phi_{2} & \lambda_{2}^{n-2} \phi_{2} & \lambda_{2}^{n-3} \varphi_{2} & \ldots & \lambda_{2} \phi_{2} & \varphi_{2} \\
\vdots & \vdots & \vdots & \vdots & \vdots & \vdots \\
\lambda_{n}^{n} \phi_{n} & \lambda_{n}^{n-2} \phi_{n} & \lambda_{n}^{n-3} \varphi_{n} & \ldots & \lambda_{n} \phi_{n} & \varphi_{n}
\end{array}\right|,
\ \\
$$
$$
W_{21}=\left|\begin{array}{cccccc}
\lambda_{1}^{n-1} \phi_{1} & \lambda_{1}^{n-2} \varphi_{1} & \lambda_{1}^{n-3} \phi_{1} & \ldots & \lambda_{1} \varphi_{1} & \phi_{1} \\
\lambda_{2}^{n-1} \phi_{2} & \lambda_{2}^{n-2} \varphi_{2} & \lambda_{2}^{n-3} \phi_{2} & \ldots & \lambda_{2} \varphi_{2} & \phi_{2} \\
\vdots & \vdots & \vdots & \vdots & \vdots & \vdots \\
\lambda_{n}^{n-1} \phi_{n} & \lambda_{n}^{n-2} \varphi_{n} & \lambda_{n}^{n-3} \phi_{n} & \ldots & \lambda_{n} \varphi_{n} & \phi_{n}
\end{array}\right|. \\
\ \\
$$
\end{theorem}

\section{Breathers on the periodic and double-periodic background}
\ \ \ \
Starting from the seed solutions $q(x,t)=ce^{i(ax+bt)}$ and $q(-x,-t)=ce^{-i(ax+bt)}$, where $b=-ac^{2}+a^{2}$ and $c$ denoting the background height. In this section, the periodic solution, breathers on the periodic background are given by odd-fold DT. In addition, the double-periodic solution, breathers and breathers on the double-periodic background solution are given by even-fold DT.

Solving the Lax pair \eqref{xlax} and \eqref{tlax}, then it gives
$$
\Psi_{1k}=\left(\begin{array}{c}
            \psi_{11k} \\
            \psi_{12k}
          \end{array}\right)=\left(\begin{array}{l}
\mathrm{e}^{\frac{1}{2} t R(-c^{2}+2\lambda_{k}^{2}+a)+\frac{1}{2} x R+\frac{1}{2}(i(ax+bt))} \\
\frac{ia-2i \lambda^{2}_{k}+R}{2c\lambda_{k}}\mathrm{e}^{\frac{1}{2} t R (-c^{2}+2 \lambda_{k}^{2}+a)+\frac{1}{2} x R-\frac{1}{2}(i(ax+bt))}
\end{array}\right),
$$
$$
\Psi_{2k}=\left(\begin{array}{c}
            \psi_{21k} \\
            \psi_{22k}
          \end{array}\right)=\left(\begin{array}{l}
\mathrm{e}^{-\frac{1}{2} t R(-c^{2}+2 \lambda_{k}^{2}+a)-\frac{1}{2} x R \lambda_{k}+\frac{1}{2}(i(ax+bt))} \\
\frac{ia-2i \lambda_{k}^{2}-R}{2c\lambda_{k}}\mathrm{e}^{-\frac{1}{2} tR(-c^{2}+2 \lambda_{k}^{2}+a)-\frac{1}{2} x R-\frac{1}{2}(i(ax+bt))}
\end{array}\right),
$$
where $$
R=\sqrt{-4c^{2}\lambda_{k}^{2}-4\lambda_{k}^{4}+4a\lambda_{k}^{2}-a^{2}}.
$$

In order to obtain nontrivial solution of Eq. \eqref{ndnls}, we constructed the new eigenfunctions associated with $\lambda_{k}$ by the linear  superposition principle.
\begin{equation}\label{rs}
\begin{split}
\phi_{k}=\psi_{11k}+\psi_{21k}+\psi_{12k}(-x, -t)+\psi_{22k}(-x, -t),\\
\psi_{k}=\psi_{12k}+\psi_{22k}+\psi_{11k}(-x, -t)+\psi_{21k}(-x, -t).
\end{split}
\end{equation}
Using eigenfunction \eqref{rs} and symmetry condition \eqref{sc}, we can get some interesting new solutions.

When $n=1$: For simple, let $\lambda_{1}=i \beta_{1},$ then
\begin{equation}\label{qf1}
|q[1]|^{2}
=\frac{[(iR_{1}-\omega_{1})e^{\omega_{2}+\omega_{5}}-(iR_{1}+\omega_{1})e^{\omega_{5}}]
       [(\omega_{3}+\omega_{4})e^{\omega_{2}}+\omega_{3}-\omega_{4}]}
       {[(iR_{1}+\omega_{1})e^{\omega_{2}}-iR_{1}+\omega_{1}]^{2}},\\
\end{equation}
where
\begin{equation}
\begin{split}
&\omega_{1}=2\beta_{1}^{2}+2\beta_{1}c+a,\\
&\omega_{2}=R_{1}[(-2\beta_{1}^{2}-c^{2}+a)t+x],\\
&\omega_{3}=4\beta_{1}^{3}+2\beta_{1}^{2}c+(2a-2c^{2})\beta_{1}-ac,\\
&\omega_{4}=iR_{1}(2\beta_{1}+c),\\
&\omega_{5}=iac^{2}t-ia^{2}t-iax,\\
&R_{1}=\sqrt{4c^{2}\beta_{1}^{2}-\left(2\beta_{1}^{2}+a\right)^{2}}.
\end{split}
\end{equation}

$\lim\limits_{\substack{x \rightarrow \infty \\ t\rightarrow \infty}}|q[1]|^{2}=c^{2}-2a$, and the trajectory of the solution \eqref{qf1} is $x=(2\beta_{1}^{2}+c^{2}-a)t$. According to these results, we can control the structure of the solution \eqref{qf1} by adjusting the parameters.

\begin{enumerate}
   \item  When $c^{2}>2a$, Eq. \eqref{qf1} can generate soliton solutions. More profound, we find that Eq. \eqref{qf1} can generate a dark soliton when $c^{2}>2a>0$, $\frac{c}{2}+\frac{\sqrt{c^{2}-2a}}{2}>\beta_{1}>\frac{c}{2}-\frac{\sqrt{c^{2}-2a}}{2}$ or $c^{2}>0>2a$, $\frac{c}{2}+\frac{\sqrt{c^{2}-2a}}{2}>\beta_{1}>-\frac{c}{2}+\frac{\sqrt{c^{2}-2a}}{2}$ (see Fig. \ref{ds}).
 Eq. \eqref{qf1} can generate a bright soliton  if $c^{2}>2a>0$, $-\frac{c}{2}+\frac{\sqrt{c^{2}-2a}}{2}>\beta_{1}>-\frac{c}{2}-\frac{\sqrt{c^{2}-2a}}{2}$ or  $c^{2}>0>2a$, $\frac{c}{2}-\frac{\sqrt{c^{2}-2a}}{2}>\beta_{1}>-\frac{c}{2}-\frac{\sqrt{c^{2}-2a}}{2}$ (see Fig. \ref{bs}).
 Eq. \eqref{qf1} can generate periodic solutions when $\beta_{1}$ belongs to other intervals.
 \item $c^{2}\leq 2a$, $\forall \beta_{1}\in R$, Eq. \eqref{qf1} also generate periodic solution (see Fig. \ref{2ps}).
\end{enumerate}

\begin{figure}[ht!]
\centering
\subfigure[Dark soliton]{
\label{ds}
\includegraphics[width=4.8cm]{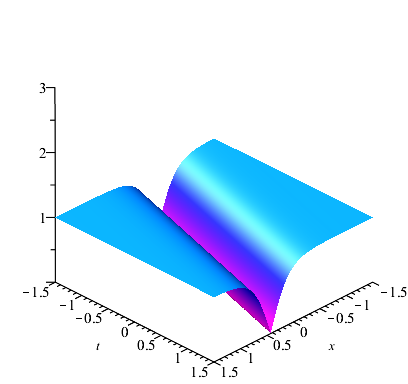}}
\subfigure[Bright soliton]{
\label{bs}
\includegraphics[width=4.2cm]{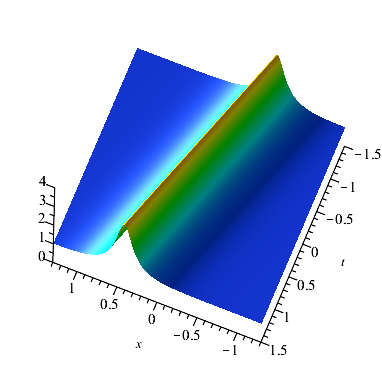}}
\subfigure[Periodic solution]{
\label{2ps}
\includegraphics[width=3.8cm]{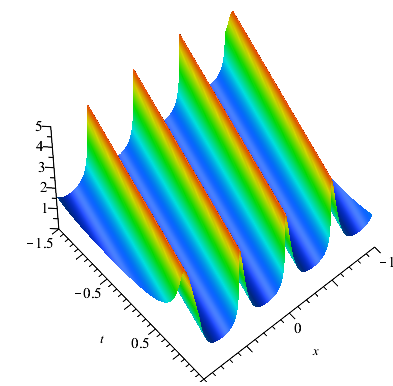}}
\caption{The dynamics of dark soliton, bright soliton and periodic solution generated from  non-zero seed solution: (a) $a=1$, $c=\sqrt{3}$, $\beta_{1}=-\frac{\sqrt{3}}{2}$;  (b) $a = 1$, $c= \sqrt{3}$, $\beta_{1}$ = $\frac{\sqrt{3}}{2}$; (c) $a = 1$, $c = \sqrt{3}$, $\beta_{1}$ = $\frac{3}{2}$.}
\label{fig2}
\end{figure}

\ \ \\

For n=2: The two-fold DT formula \eqref{dtf} of the reverse-space-time DNLS equation implies a solution
\begin{equation}\label{q2b}
q[2]=\frac{-2[i(\lambda_{1}^{2}-\lambda_{2}^{2})\phi_{1}\phi_{2}-\frac{1}{2}\lambda_{2}\phi_{1}\psi_{2}
+\frac{1}{2}\lambda_{1}q\phi_{2}\psi_{1}](-\lambda_{1}\phi_{2}\psi_{1}+\lambda_{2}\phi_{1}\psi_{2})}
{(\lambda_{1}\phi_{1}\psi_{2}-\lambda_{2}\phi_{2}\psi_{1})^{2}},
\end{equation}
where
\begin{equation}
 \left(\begin{array}{c}
            \phi_{2} \\
            \psi_{2}
  \end{array}\right)=
  \left(\begin{array}{c}
            \psi_{1}(-x,-t;\lambda_{2}) \\
            \phi_{1}(-x,-t;\lambda_{2})
  \end{array}\right).
 \end{equation}
We can derive breathers and double-periodic solution according to different reduced methods of spectrum parameter $\lambda_{1}$ and $\lambda_{2}$ as follows.

Case 1: $\lambda_{2}=\pm\lambda_{1}^{*}$,  now $q[2]$ is a breathers. For simplicity, we take $\lambda_{2}=-\lambda_{1}^{*}=-\alpha_{1}+i \beta_{1}$ and $ \operatorname{Im}\left(-a^{2}-4 \lambda_{1}^{4}-4 \lambda_{1}^{2}\left(c^{2}-a\right)\right)=0$. Then $\lim\limits_{\substack{x \rightarrow \infty \\ t\rightarrow \infty}}|q[2]|^{2}=c^{2}$.
      and the center trajectory equation of solution $q[2]$ can be calculated as $x=4(\beta_{1}^{2}-\alpha_{1}^{2})t $. We can know that the solution evolves periodically along the straight line with a certain angle of $x$ axis and $t$ axis when $\beta_{1}^{2}\neq\alpha_{1}^{2}$ from the above analysis (see Fig. \ref{ghxt}). And when $\beta_{1}^{2}=\alpha_{1}^{2}$, the classical Ma breathers (time periodic breather) can be seen in  Fig. \ref{mhx} and the Akhmediev breathers (space periodic breather) can be seen in Fig. \ref{akhx}.
     \begin{figure}[ht!]
\centering
\subfigure[General breathers]{
\label{ghxt}
\includegraphics[width=3.5cm]{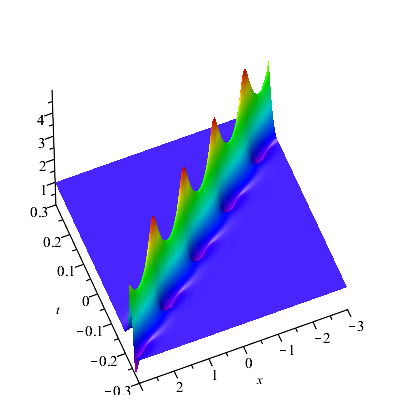}}
\subfigure[Ma breathers]{
\label{mhx}
\includegraphics[width=3.4cm]{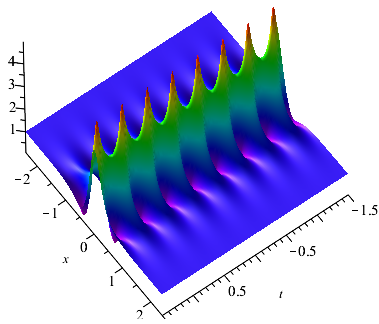}}
\subfigure[Akmediev breathers]{
\label{akhx}
\includegraphics[width=3.3cm]{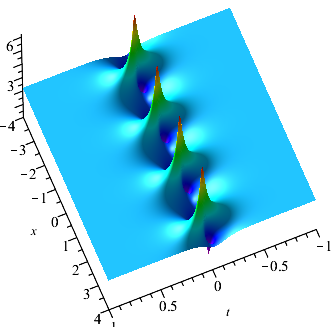}}
\subfigure[Rogue wave]{
\label{gbt}
\includegraphics[width=3.4cm]{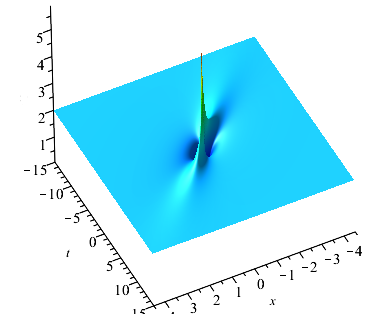}}
\caption{The dynamics of breathers and rogue waves with: (a) $a = 7$, $c = 1$, $\alpha_{1}=2$, $\beta_{1}=1$; (b) $a = 1$, $c = 1$, $\alpha_{1}=1$, $\beta_{1}=1$; (c) $a = 5$, $c = \sqrt{5}$, $\alpha_{1}=1$, $\beta_{1}=1$; (d) $\alpha_{1}=\frac{1}{2}$, $\beta_{1} = 1$; }
\end{figure}

Next, we construct rogue wave $q_{r}$ by letting the period of the breathers tend to be infinity. Let $c\rightarrow-2\beta_{1}$ in $q[2]$, then
\begin{equation}\label{00gbj}
|q_{r}|^{2}=\frac{m_{1}t^{4}+m_{2}t^{3}+m_{3}t^{2}
 +m_{4}t+m_{5}}{m_{6}t^{4}+m_{7}t^{3}+m_{8}t^{2}+m_{9}t+m_{10}},
\end{equation}
\begin{equation}\notag
\begin{split}
&m_{1}=262144(\alpha_{1}^{12}\beta_{1}^{6}+2\alpha_{1}^{6}\beta_{1}^{12}+\beta_{1}^{18}),\\
&m_{2}=262144(\alpha_{1}^{10}\beta_{1}^{6}-\alpha_{1}^{6}\beta_{1}^{10} +\alpha_{1}^{4}\beta_{1}^{12}-\beta_{1}^{16})x,\\
&m_{3}=
(32768(3\alpha_{1}^{8}\beta_{1}^{6}+\alpha_{1}^{6}\beta_{1}^{8}-4\alpha_{1}^{4}\beta_{1}^{10}+\alpha_{1}^{2}\beta_{1}^{12}+3\beta_{1}^{14})x^{2}
-6144(\alpha_{1}^{6}\beta_{1}^{4}-6\alpha_{1}^{4}\beta_{1}^{6}+\beta_{1}^{10})),\\
&m_{4}=(16384(\alpha_{1}^{6}\beta_{1}^{6}+\alpha_{1}^{4}\beta_{1}^{8}-\alpha_{1}^{2}\beta_{1}^{10}-\beta_{1}^{12})x^{3}
 +3072(2\alpha_{1}^{2}\beta_{1}^{6}+\beta_{1}^{8}-\alpha_{1}^{4}\beta_{1}^{4})x),\\
&m_{5}=1024(\alpha_{1}^{4}\beta_{1}^{6}+2\alpha_{1}^{2}\beta_{1}^{8}+\beta_{1}^{10})x^{4}-128(3\alpha_{1}^{2}\beta_{1}^{4}
 +\beta_{1}^{6})x^{2}+36\beta_{1}^{2},\\
&m_{6}=65536(\alpha_{1}^{12}\beta_{1}^{4}+2\alpha_{1}^{6}\beta_{1}^{10}+\beta_{1}^{16}),\\
&m_{7}=65536(\alpha_{1}^{10}\beta_{1}^{4}-\alpha_{1}^{6}\beta_{1}^{8}+\alpha_{1}^{4}\beta_{1}^{10}-\beta_{1}^{14})x,\\
&m_{8}=8192(3\alpha_{1}^{8}\beta_{1}^{4}+\alpha_{1}^{6}\beta_{1}^{6}
 -4\alpha_{1}^{4}\beta_{1}^{8}+\alpha_{1}^{2}\beta_{1}^{10}+3\beta_{1}^{12})x^{2}
 +512(\alpha_{1}^{6}\beta_{1}^{2}+2\alpha_{1}^{4}\beta_{1}^{4}-8\alpha_{1}^{2}\beta_{1}^{6}+9\beta_{1}^{8}),\\
&m_{9}=4096(\alpha_{1}^{6}\beta_{1}^{4}+\alpha_{1}^{4}\beta_{1}^{6}-\alpha_{1}^{2}\beta_{1}^{8}-\beta_{1}^{10})x^{3}
+256(\alpha_{1}^{4}\beta_{1}^{2}+2\alpha_{1}^{2}\beta_{1}^{4}-5\beta_{1}^{6})x,\\
&m_{10}=256(\alpha_{1}^{4}\beta_{1}^{4}+2\alpha_{1}^{2}\beta_{1}^{6}+\beta_{1}^{8})x^{4}+32(\alpha_{1}^{2}\beta_{1}^{2}
 +3\beta_{1}^{4})x^{2}+1.\\
 \end{split}
\end{equation}

After calculation and analysis, we know that $\lim\limits_{\substack{x \rightarrow \infty \\ t\rightarrow \infty}}|q_{r}|^{2}=(2\beta_{1})^{2}$. The maximum amplitude of $|q_{r}|^{2}$ equals to $(6\beta_{1})^{2}$ occurs at $x=0$ and $t=0$. This means that the maximum amplitude of the rogue wave $q_{r}$ is $3$ times compared with the asymptotic plane wave at infinity. The min amplitude of $|q_{r}|^{2}$ occurs at ($\sqrt{\frac{27\alpha_{1}^{4}}{16\beta_{1}^{2}(4\alpha_{1}^{2}+\beta_{1}^{2})(\alpha_{1}^{2}+\beta_{1}^{2})^{2}}}$,
$\sqrt{\frac{3}{256\beta_{1}^{2}(4\alpha_{1}^{2}+\beta_{1}^{2})(\alpha_{1}^{2}+\beta_{1}^{2})^{2}}}$) and ($-\sqrt{\frac{27\alpha_{1}^{4}}{16\beta_{1}^{2}(4\alpha_{1}^{2}+\beta_{1}^{2})(\alpha_{1}^{2}+\beta_{1}^{2})^{2}}}$,
$-\sqrt{\frac{3}{256\beta_{1}^{2}(4\alpha_{1}^{2}+\beta_{1}^{2})(\alpha_{1}^{2}+\beta_{1}^{2})^{2}}}$), which  equals to $\frac{108(3\alpha_{1}^{4}-2\alpha_{1}^{2}\beta_{1}^{2}+4\beta_{1}^{4})\alpha_{1}^{4}\beta_{1}^{2}}
{169\alpha_{1}^{8}-56\alpha_{1}^{6}\beta_{1}^{2}+6\alpha_{1}^{4}\beta_{1}^{4}-8\alpha_{1}^{2}\beta_{1}^{6}+4\beta_{1}^{8}}$. The rogue wave $|q_{r}|$ with specific parameter $\alpha_{1}=\frac{1}{2}$ and $\beta_{1} = 1$ is plotted in Fig. \eqref{gbt}.
From the graph of the rogue wave $|q_{r}|$ , we can see that the rogue wave  $q_{r}$ has a single peak with two caves on both sides of the peak. The optical pulse $q_{r}$ only exists locally with all variables and disappears as time and space go far.

Case 2: $\lambda_{1}=i\beta_{1}$, $\lambda_{2}=i\beta_{2}$ and $\beta_{2}\neq\pm\beta_{1}$, $q[2]$ is represented as a double-periodic wave solution which is similar to the Jacobi elliptic function-type seed solution. From the Fig. \ref{shxt}, we can  see clearly that there are two periodic waves with different directions in the double-periodic wave solution, and when the two waves with different directions are superimposed on each other, a higher wave peak can be generated. From a visual perspective, it seems that several parallel breathers are generated under the period background.
\begin{figure}[ht!]
\centering
{
\includegraphics[width=4.8cm]{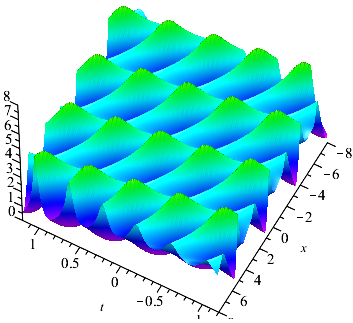}
\includegraphics[width=4.3cm]{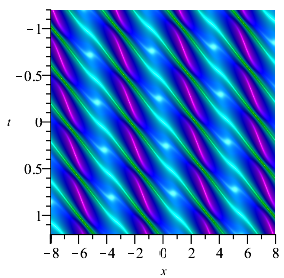}
\includegraphics[width=3.8cm]{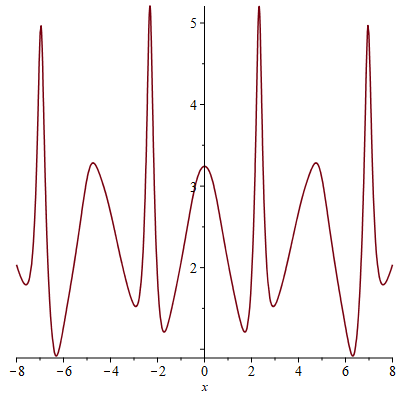}}
\caption{The dynamics of double-periodic solution with: $a = 1$, $c = 1$, $\beta_{1} = \sqrt{2}$, $\beta_{2} = \frac{\sqrt{2}}{2}$.}
\label{shxt}
\end{figure}

We find that the periodic solution can be generated by first-fold DT. The breathers and double-periodic solutions can be constructed respectively according to second-fold DT. Therefore, we consider constructing the breathers on the periodic background by odd-fold DT, and constructing the breathers on the double-periodic background by even-fold DT.

For n=3: Set $\lambda_{2}=-\lambda_{1}^{*}=-\alpha_{1}+i \beta_{1}$, $\lambda_{3}= i\beta_{3}$ and $a=2\alpha_{1}^{2}-2\beta_{1}^{2}+c^{2}$. Parameter values have a great influence on the propagation direction of the breathers. When $\beta_{1}^{2}\neq\alpha_{1}^{2}$,  solution $q[3]$ is a general breather solution on periodic background (see Fig. \ref{xthxz}). When $\beta_{1}^{2}=\alpha_{1}^{2}$, the classical Ma breathers on the periodic background can be seen in  Fig. \ref{xhxz} and the Akhmediev breathers on the periodic background can be seen in Fig. \ref{thxz}. There are some interesting phenomenons: Under the perturbation of the periodic background, the crest of the Ma breathers is cut and the phase shift occurs at the center of the breathers.
\begin{figure}[ht!]
\centering
\subfigure[General breathers on the periodic background]{\label{xthxz}
\begin{minipage}[b]{0.3\textwidth}
\includegraphics[width=4.2cm]{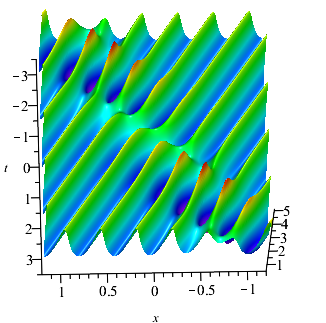} \\
\end{minipage}
}
\subfigure[Ma breathers on the periodic background]{\label{xhxz}
\begin{minipage}[b]{0.3\textwidth}
\includegraphics[width=4.2cm]{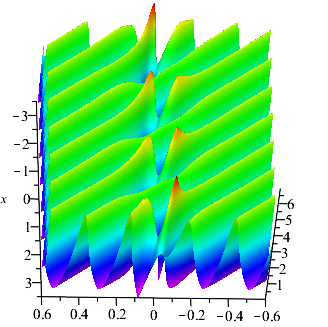} \\
\end{minipage}
}
\subfigure[Akhmediev breathers on the periodic background]{\label{thxz}
\begin{minipage}[b]{0.3\textwidth}
\includegraphics[width=4.4cm]{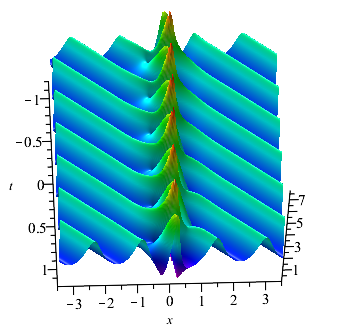} \\
\end{minipage}
}
\caption{The dynamics of breathers on the periodic background: (a) $a=1$, $c=1$, $\alpha_{1}=1$, $\beta_{1}=\frac{1}{3}$, $\beta_{3}=\sqrt{2}$; (b) $a=1$, $c=1$, $\alpha_{1}=1$, $\beta_{1}=1$, $\beta_{3}=\sqrt{2}$; (c) $a=5$, $c=\sqrt{5}$, $\alpha_{1}=1$, $\beta_{1}=1$, $\beta_{3}=\sqrt{2}$;}
\label{q3t}
\end{figure}
\ \ \ \\

For n=4:  Set $\lambda_{2}=-\lambda_{1}^{*}=-\alpha_{1}+i \beta_{1}$, $\lambda_{3}= i\beta_{3}$, $\lambda_{4}= i\beta_{4}$ and $\beta_{4}\neq\pm\beta_{3}$. The breathers on a double-periodic background generated by formula \eqref{dtf}.  Similar to $n=2$, we also let $a=2 \alpha_{1}^{2}-2 \beta_{1}^{2}+c^{2}$.  When $\beta_{1}^{2}\neq\alpha_{1}^{2}$, we can construct the general breathers on the double-periodic (see Fig. \ref{q4xt}). Under the disturbance of double-periodic background, the propagation direction of general breathers usually produces shift.  When $\beta_{1}^{2}=\alpha_{1}^{2}$, the Ma breathers and Akhmediev breathers can be constructed by adjusting spectrum parameters. As for the Ma breathers on the double-periodic, due to the great influence of the double-periodic background, the image of Ma breathers solution looks like it disappears in the double-periodic background (see Fig. \ref{q4t6}).
The Akhmediev breathers on the double-periodic background is plotted in Fig. \ref{q4}. Visually, it looks like a breathers with a higher amplitude is generated under the several parallel breathers background.

\begin{figure}[ht!]
\centering
\subfigure[General breathers on the double-periodic background]{
\label{q4xt}
\begin{minipage}[b]{0.3\textwidth}
\includegraphics[width=3.8cm]{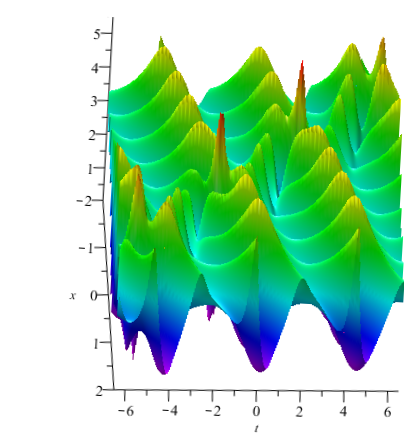} \\
\includegraphics[width=4.0cm]{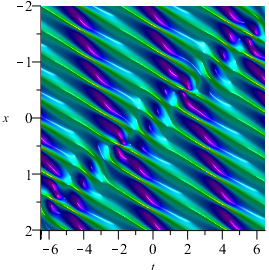}
\end{minipage}
}
\subfigure[Ma breathers on the double-periodic background]{
\label{q4t6}
\begin{minipage}[b]{0.3\textwidth}
\includegraphics[width=4.2cm]{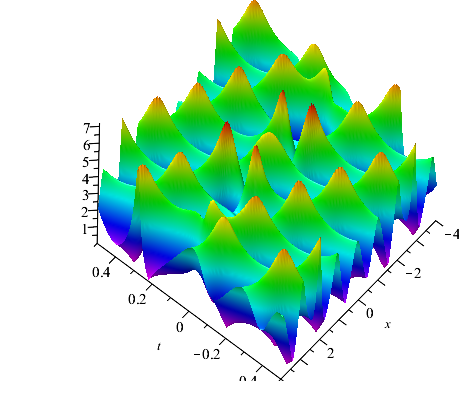} \\
\includegraphics[width=4.0cm]{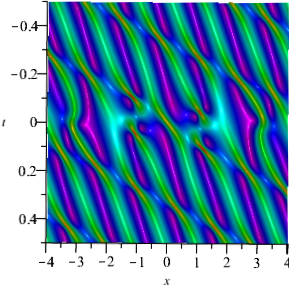}
\end{minipage}
}
\subfigure[Akhmediev Breathers on the double-periodic background]{
\label{q4}
\begin{minipage}[b]{0.3\textwidth}
\includegraphics[width=4.5cm]{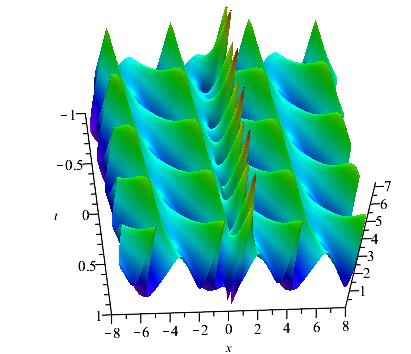}\\
\includegraphics[width=4.0cm]{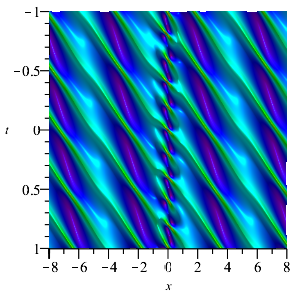}
\end{minipage}
}
\caption{The dynamics of breathers on the double-periodic background with: (a) $a=\frac{7}{9}$, $c=1$, $\alpha_{1}=1$, $\beta_{1}=\frac{1}{3}$, $\beta_{3}=\sqrt{2}$, $\beta_{4}=\frac{\sqrt{2}}{2}$; (b) $a=5$, $c=\sqrt{5}$, $\alpha_{1}=1$, $\beta_{1}=1$, $\beta_{3}=\sqrt{2}$, $\beta_{4}=\frac{\sqrt{2}}{2}$; (c) $a=1$, $c=1$, $\alpha_{1}=1$, $\beta_{1}=1$, $\beta_{3}=\sqrt{2}$, $\beta_{4}=\frac{\sqrt{2}}{2}$. }
\label{1-qb-periodic}
\end{figure}
In addition, if we  set $\lambda_{2}=-\lambda_{1}^{*}=-\alpha_{1}+i \beta_{1}$ and $\lambda_{4}=-\lambda_{3}^{*}=-\alpha_{3}+i \beta_{3}$.
When $\beta_{3}^{2}-\alpha_{3}^{2}$=$\beta_{1}^{2}-\alpha_{1}^{2}$, we can  construct the velocity resonance of two pairs of breathers (see Fig. \ref{n4hg} ). Otherwise, we can construct the elastic collision  of two pairs of breathers (see Fig. \ref{n4tp} ).
\begin{figure}[ht!]
\centering
\subfigure[Velocity resonance of breathers]{\label{n4hg}
\begin{minipage}[b]{0.35\textwidth}
\includegraphics[width=4.5cm]{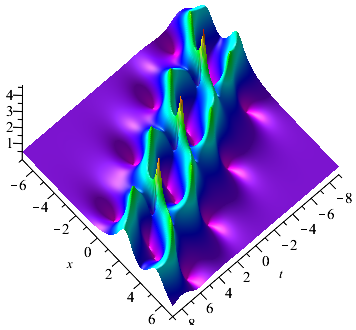}
\end{minipage}}
\subfigure[Elastic collision of breathers]{\label{n4tp}
\begin{minipage}[b]{0.35\textwidth}
\includegraphics[width=4.5cm]{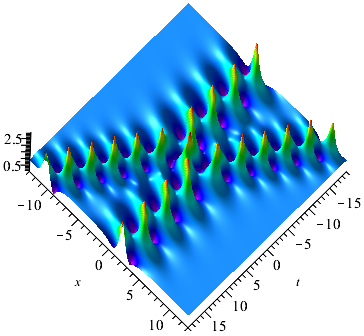}
\end{minipage}}
\caption{The dynamics of velocity resonance and elastic collision of breathers with: (a) $a=\frac{7}{25}$, $c=\frac{1}{2}$, $\alpha_{1}=\frac{1}{2}$, $\beta_{1}=\frac{3}{5}$, $\alpha_{3}=\frac{3}{5}$, $\beta_{3}=\frac{\sqrt{47}}{10}$; (b) $a=\frac{361}{800}$, $c=\frac{19}{20}$, $\alpha_{1}=-\frac{1}{2}$, $\beta_{1}=\frac{1}{2}$, $\alpha_{3}=\frac{3}{5}$, $\beta_{3}=\frac{1}{2}$.}
\label{q4gb}
\end{figure}
\section{Higher-order Rogue waves on the double-periodic background}

\ \ \ \
Note that the eigenfunction is degenerated when $\lambda=\pm\frac{1}{2}\sqrt{2a-c^{2}}-\frac{1}{2}ic$. In this case, the higher-order rogue waves can be derived. Combined with the methods of constructing the periodic and double-periodic background in the previous section, we will give the solutions of higher-order rogue waves on the periodic and double-periodic background in this section.
 Since periodic can be derived by odd-fold DT, both breathers and double-periodic solution can be obtained by even-fold DT. We can construct higher-order rogue waves on the periodic by odd-order semi-degenerate DT and higher-order rogue waves on the double-periodic by even-order semi-degenerate DT.
 {\bf \subsection{Semi-degenerate DT formula}}

\begin{theorem} Let $\lambda_{2}=-\lambda_{1}^{*}=-\frac{1}{2}\sqrt{2a-c^{2}}-\frac{1}{2}ic$.  When $n=2k$, set $\lambda_{n-1}= i\beta_{n-1}$,  $\lambda_{n}= i\beta_{n}$ and $\beta_{n-1} \neq \pm \beta_{n}$. When $n=2k+1$, set $\lambda_{n}= i\beta_{n}$, then the semi-degenerate DT formula can be obtained as follows.
\begin{equation}\label{gdtf}
q_{n}=\frac{W_{11}^{'2}}{W_{21}^{'2}} q+2 i \frac{W_{11}^{'} W_{12}^{'}}{W_{21}^{'2}}.
\end{equation}

When $n=2k$,
$$
\begin{array}{l}
W_{11}^{'}=\left|\begin{array}{cccccc}
\varphi_{[1,n-1,1]} & \phi_{[1,n-2,1]} & \varphi_{[1,n-3,1]} & \ldots &\varphi_{[1,1,1]} & \phi_{[1,0,1]} \\
\varphi_{[2,n-1,1]} &\phi_{[2,n-2,1]} & \varphi_{[2,n-3,1]} & \ldots & \varphi_{[2,1,1]} & \phi_{[2,0,1]} \\
\vdots & \vdots & \vdots & \vdots & \vdots & \vdots \\
\varphi_{[1,n-1,k-1]} & \phi_{[1,n-2,k-1]} & \varphi_{[1,n-3,k-1]} & \ldots &\varphi_{[1,1,k-1]} & \phi_{[1,0,k-1]} \\
\varphi_{[2,n-1,k-1]} &\phi_{[2,n-2,k-1]} & \varphi_{[2,n-3,k-1]} & \ldots & \varphi_{[2,1,k-1]} & \phi_{[2,0,k-1]} \\
\lambda_{n-1}^{n-1}\varphi_{n-1} & \lambda_{n-1}^{n-2}\phi_{n-1} &\lambda_{n-1}^{n-3}\varphi_{n-1} & \ldots & \lambda_{n-1}\varphi_{n-1} & \phi_{n-1}\\
\lambda_{n}^{n-1}\varphi_{n} & \lambda_{n}^{n-2}\phi_{n} &\lambda_{n}^{n-3}\varphi_{n} & \ldots & \lambda_{n}\varphi_{n} & \phi_{n}
\end{array}\right|, \\
\ \\
\end{array}
$$
$$
\begin{array}{l}
W_{12}^{'}=\left|\begin{array}{cccccc}
 \phi_{[1,n,1]} & \phi_{[1,n-2,1]} & \varphi_{[1,n-3,1]} & \ldots & \varphi_{[1,1,1]} & \phi_{[1,0,1]} \\
\phi_{[2,n,1]} &  \phi_{[2,n-2,1]} & \varphi_{[2,n-3,1]} & \ldots & \varphi_{[2,1,1]]} & \phi_{[2,0,1]} \\
\vdots & \vdots & \vdots & \vdots & \vdots & \vdots \\
 \phi_{[1,n,k-1]} & \phi_{[1,n-2,k-1]} & \varphi_{[1,n-3,k-1]} & \ldots & \varphi_{[1,1,k-1]} & \phi_{[1,0,k-1]} \\
\phi_{[2,n,k-1]} &  \phi_{[2,n-2,k-1]} & \varphi_{[2,n-3,k-1]} & \ldots & \varphi_{[2,1,k-1]]} & \phi_{[2,0,k-1]} \\
\lambda_{n-1}^{n}\phi_{n-1} & \lambda_{n-1}^{n-2}\phi_{n-1} &\lambda_{n-1}^{n-3}\varphi_{n-1} & \ldots & \lambda_{n-1}\varphi_{n-1} & \phi_{n-1}\\
\lambda_{n}^{n}\phi_{n} & \lambda_{n}^{n-2}\phi_{n} &\lambda_{n}^{n-3}\varphi_{n} & \ldots & \lambda_{n}\varphi_{n} & \phi_{n}
\end{array}\right|, \\
\ \\
\end{array}
$$
$$
\begin{array}{l}
W_{21}^{'}=\left|\begin{array}{cccccc}
\phi_{[1,n-1,1]} & \varphi_{[1,n-2,1]} & \phi_{[1,n-3,1]} & \ldots & \phi_{[1,1,1]} & \varphi_{[1,0,1]} \\
\phi_{[2,n-1,1]} & \varphi_{[2,n-2,1]} & \phi_{[2,n-3,1]} & \ldots &\phi_{[2,1,1]} & \varphi_{[2,0,1]} \\
\vdots & \vdots & \vdots & \vdots & \vdots & \vdots \\
\phi_{[1,n-1,k-1]} & \varphi_{[1,n-2,k-1]} & \phi_{[1,n-3,k-1]} & \ldots & \phi_{[1,1,k-1]} & \varphi_{[1,0,k-1]} \\
\phi_{[2,n-1,k-1]} & \varphi_{[2,n-2,k-1]} & \phi_{[2,n-3,k-1]} & \ldots &\phi_{[2,1,k-1]} & \varphi_{[2,0,k-1]} \\
\lambda_{n-1}^{n-1}\phi_{n} & \lambda_{n-1}^{n-2}\varphi_{n} &\lambda_{n-1}^{n-3}\phi_{n} & \ldots & \lambda_{n-1}\phi_{n} & \varphi_{n-1}\\
\lambda_{n}^{n-1}\phi_{n} & \lambda_{n}^{n-2}\varphi_{n} &\lambda_{n}^{n-3}\phi_{n} & \ldots & \lambda_{n}\phi_{n} & \varphi_{n}
\end{array}\right|, \\
\ \\
\end{array}
$$

When $n=2k+1,$
$$
\begin{array}{l}
W_{11}^{'}=\left|\begin{array}{cccccc}
\varphi_{[1,n-1,1]} & \phi_{[1,n-2,1]} & \varphi_{[1,n-3,1]} & \ldots &\phi_{[1,1,1]} & \varphi_{[1,0,1]} \\
\varphi_{[2,n-1,1]} &\phi_{[2,n-2,1]} & \varphi_{[2,n-3,1]} & \ldots & \phi_{[2,1,1]} & \varphi_{[2,0,1]} \\
\vdots & \vdots & \vdots & \vdots & \vdots & \vdots \\
\varphi_{[1,n-1,k]} & \phi_{[1,n-2,k]} & \varphi_{[1,n-3,k]} & \ldots &\phi_{[1,1,k]} & \varphi_{[1,0,k]} \\
\varphi_{[2,n-1,k]} &\phi_{[2,n-2,k]} & \varphi_{[2,n-3,k]} & \ldots & \phi_{[2,1,k]} & \varphi_{[2,0,k]} \\
\lambda_{n}^{n-1}\varphi_{n} & \lambda_{n}^{n-2}\phi_{n} & \lambda_{n}^{n-3}\varphi_{n} & \ldots & \lambda_{n}\phi_{n} & \varphi_{n}
\end{array}\right|, \\
\ \\
\end{array}
$$

$$
\begin{array}{l}
W_{12}^{'}=\left|\begin{array}{cccccc}
 \phi_{[1,n,1]} & \phi_{[1,n-2,1]} & \varphi_{[1,n-3,1]} & \ldots & \phi_{[1,1,1]} & \varphi_{[1,0,1]} \\
\phi_{[2,n,1]} &  \phi_{[2,n-2,1]} & \varphi_{[2,n-3,1]} & \ldots & \phi_{[2,1,1]]} & \varphi_{[2,0,1]} \\
\vdots & \vdots & \vdots & \vdots & \vdots & \vdots \\
 \phi_{[1,n,k]} & \phi_{[1,n-2,k]} & \varphi_{[1,n-3,k]} & \ldots & \phi_{[1,1,k]} & \varphi_{[1,0,k]} \\
\phi_{[2,n,k]} &  \phi_{[2,n-2,k]} & \varphi_{[2,n-3,k]} & \ldots & \phi_{[2,1,k]]} & \varphi_{[2,0,k]} \\
\lambda_{n}^{n}\phi_{n} & \lambda_{n}^{n-2}\phi_{n}&  \lambda_{n}^{n-3}\varphi_{n} & \ldots & \lambda_{n}\phi_{n} & \varphi_{n}
\end{array}\right|, \\
\ \\
\end{array}
$$

$$
\begin{array}{l}
W_{21}^{'}=\left|\begin{array}{cccccc}
\phi_{[1,n-1,1]} & \varphi_{[1,n-2,1]} & \phi_{[1,n-3,1]} & \ldots & \varphi_{[1,1,1]} & \phi_{[1,0,1]} \\
\phi_{[2,n-1,1]} & \varphi_{[2,n-2,1]} & \phi_{[2,n-3,1]} & \ldots &\varphi_{[2,1,1]} & \phi_{[2,0,1]} \\
\vdots & \vdots & \vdots & \vdots & \vdots & \vdots \\
\phi_{[1,n-1,k]} & \varphi_{[1,n-2,k]} & \phi_{[1,n-3,k]} & \ldots & \varphi_{[1,1,k]} & \phi_{[1,0,k]} \\
\phi_{[2,n-1,k]} & \varphi_{[2,n-2,k]} & \phi_{[2,n-3,k]} & \ldots &\varphi_{[2,1,k]} & \phi_{[2,0,k]} \\
\lambda_{n}^{n-1}\phi_{n} & \lambda_{n}^{n-2}\varphi_{n}&  \lambda_{n}^{n-3}\phi_{n} & \ldots & \lambda_{n}\varphi_{n} & \phi_{n}
\end{array}\right|. \\
\ \\
\end{array}
$$
\end{theorem}

\begin{proof}
Define the new function $\Psi[i, j, k]$ as follows
\begin{equation}
\label{te}
\lambda^{j} \Psi=\Psi[i, j, 0]+\Psi[i, j, 1] \varepsilon+\Psi[i, j, 2] \varepsilon^{2}+\cdots+\Psi[i, j, k] \varepsilon^{k}+\cdots,
\end{equation}
where $\varepsilon$ is a small parameter, $\Psi[i, j, k]=\frac{1}{k!} \frac{\partial^{k}}{\partial \varepsilon^{k}}\left[\left(\lambda_{i}+\varepsilon\right)^{j} \Psi\left(\lambda_{i}+\varepsilon\right)\right]$.

When $n=2k$, set $\lambda_{1}=\frac{1}{2}\sqrt{2a-c^{2}}-\frac{1}{2}ic+\varepsilon_{1}$, $\lambda_{2}=-\frac{1}{2}\sqrt{2a-c^{2}}-\frac{1}{2}ic+\varepsilon_{1}$, $\lambda_{2k-3}\rightarrow \lambda_{1}$ and $\lambda_{2k-2}\rightarrow \lambda_{2}$, $2\leq k\leq\frac{n}{2}$. Then using the expansion equation \eqref{te} to the elements in the first column of $W_{11}$ (remain the elements of the $(2k-1)$-th and $2k$-th row unchanged), we have
\begin{equation}\notag
\begin{aligned}
&\lambda_{1}^{n-1} \varphi_{1}=\varphi[1, n-1,1],\\
&\lambda_{2}^{n-1} \varphi_{2}=\varphi[2, n-1,1] \\
&\lambda_{3}^{n-1} \varphi_{3}=\varphi[1, n-1,1]+\varphi[1, n-1,2] \epsilon, \\
&\lambda_{4}^{n-1} \varphi_{4}=\varphi[2, n-1,1]+\varphi[2, n-1,2] \epsilon, \\
&\quad \vdots \\
&\lambda_{n-3}^{n-1} \varphi_{n-1}=\varphi[1, n-1,1]+\varphi[1, n-1,2] \epsilon+\cdots+\varphi[1, n-1, k-1] \epsilon^{k-1}, \\
&\lambda_{n-2}^{n-1} \varphi_{n}=\varphi[2, n-1,1]+\varphi[1, n-1,2] \epsilon+\cdots+\varphi[2, n-1, k-1] \epsilon^{k-1}.
\end{aligned}
\end{equation}
Taking the same procedure to the other column of $W_{11}$, $W_{12}$ and $W_{21}$. Finally, the even-order semi-degenerate DT formula $q_{n}$ can be obtained through determinant calculation. For the case of $n=2k+1$, remain all the elements of the $(2k+1)$-th row unchanged, we can derive the odd-order semi-degenerate DT formula $q_{n}$  by the same procedure of case $n=2k$.
\end{proof}

{\bf \subsection{Higher-order rogue waves on the periodic background}}

\par The Higher-order rogue waves on periodic background can be generated by odd-order semi-degenerate DT. For $n = 3$, $\lambda_{1}=\frac{1}{2}\sqrt{2a-c^{2}}-\frac{1}{2}ic$, $\lambda_{2}=-\frac{1}{2}\sqrt{2a-c^{2}}-\frac{1}{2}ic$ and $\lambda_{3}=i\beta_{3}$. Here, set $a=1$ and $c=1$ for convenience, the first-order rogue waves $q_{3}$ on the periodic background for Eq. \eqref{ndnls} can be obtained. The patterns of $q_{3}$ are displayed in Fig. \ref{zqgb1} and Fig. \ref{zqgb2}. For $\beta_{3}>0$, the rogue wave pattern locates on the area where the periodic pattern reaches its amplitude. However, for $\beta_{3}<0$, the rogue wave pattern locates in the middle of two amplitude trajectories of the periodic pattern, which looks like that the rogue wave is generated by the interaction of two waves of the periodic pattern.
\begin{figure}[ht!]
\centering
\subfigure[]{\label{zqgb1}
\begin{minipage}[b]{0.4\textwidth}
\includegraphics[width=6.0cm]{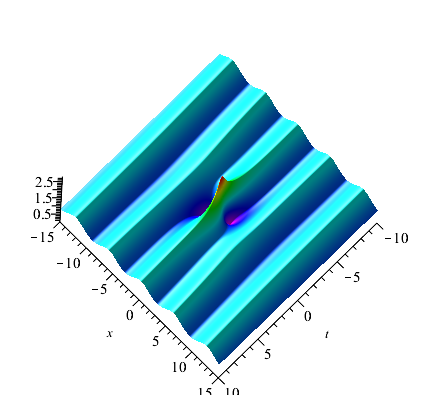} \\
\includegraphics[width=4.5cm]{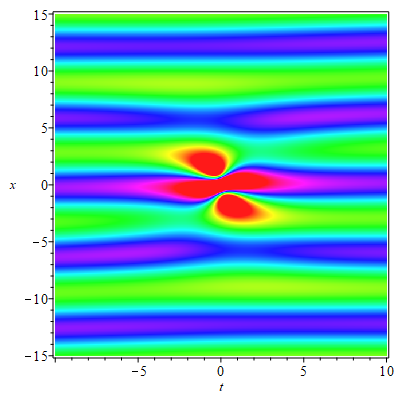}
\end{minipage}
}
\subfigure[]{\label{zqgb2}
\begin{minipage}[b]{0.4\textwidth}
\includegraphics[width=6.0cm]{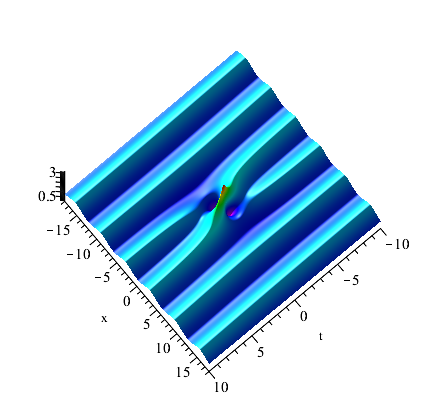} \\
\includegraphics[width=4.5cm]{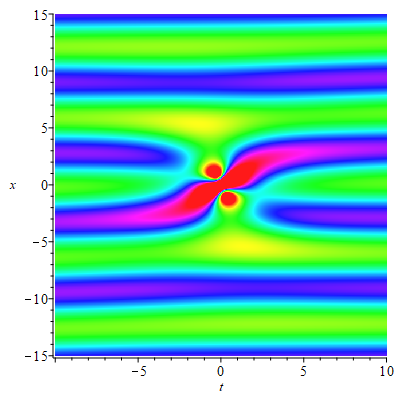}
\end{minipage}
}
\caption{The dynamics of first-order rogue waves on the periodic background: (a) $a=1$, $c=1$, $\beta_{3}=\frac{1}{10}$;
(b) $a=1$, $c=1$, $\beta_{3}=-\frac{1}{10}$. }
\label{1-qb-periodic}
\end{figure}

For $n=5$, $\lambda_{1}=\frac{1}{2}\sqrt{2a-c^{2}}-\frac{1}{2}ic$, $\lambda_{2}=-\frac{1}{2}\sqrt{2a-c^{2}}-\frac{1}{2}ic$ and $\lambda_{5}=i \beta_{5}$,  the second-order rogue waves on the periodic background for Eq. \eqref{ndnls} was shown in Fig. \ref{2-qb-periodic}. The second-order rogue waves have a high amplitude peak on the center distributed with some lower peaks and
four caves. Same with the case $n=3$, for $\beta_{5}>0$, the rogue wave pattern locates on the area where the periodic pattern reaches its amplitude (see Fig. \ref{z2qgb}). For $\beta_{5}<0$, the rogue wave pattern locates in the middle of two amplitude trajectories of the periodic pattern (see Fig. \ref{f2qgb}). And the periodic background can influence the peak value of the rogue wave.

\begin{figure}[ht!]
\centering
\subfigure[]{\label{z2qgb}
\begin{minipage}[b]{0.4\textwidth}
\includegraphics[width=6.0cm]{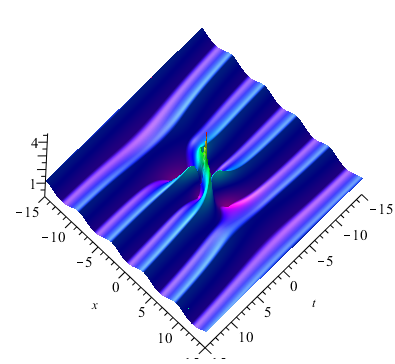} \\
\includegraphics[width=4.8cm]{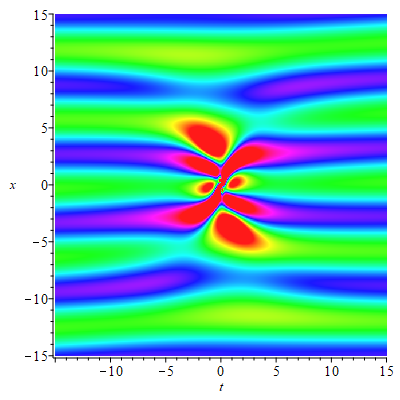}
\end{minipage}
}
\subfigure[]{\label{f2qgb}
\begin{minipage}[b]{0.4\textwidth}
\includegraphics[width=6.0cm]{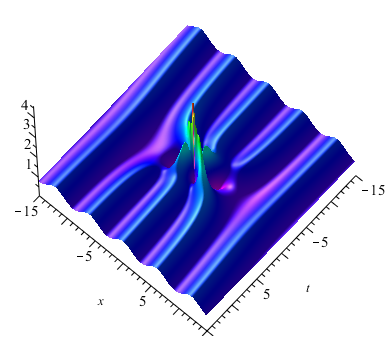} \\
\includegraphics[width=4.8cm]{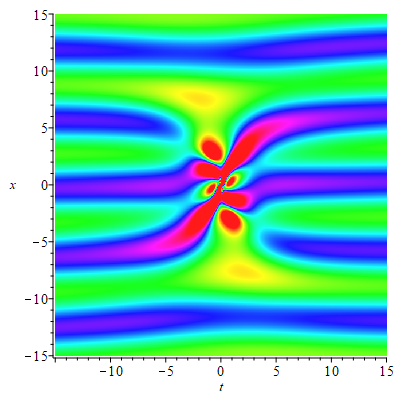}
\end{minipage}
}
\caption{Second-order rogue waves on the periodic background: (a) $a=1$, $c=1$, $\beta_{5}=\frac{1}{10}$;
(b) $a=1$, $c=1$, $\beta_{5}=-\frac{1}{10}$. }
\label{2-qb-periodic}
\end{figure}
  \ \ \ \ \\

{\bf \subsection{Higher-order rogue waves on the double-periodic background}}
\par When $n=4$,  we can obtain the first-order rogue waves on the double-periodic background by \eqref{gdtf}. The selection of parameters have effect both on the amplitude of the double-periodic background and the amplitude of the rogue waves. The interesting thing is that there are two peaks rogue wave on the double-periodic background when we take $a=1$, $c=1$, $\beta_{3}=0.1$ and $\beta_{4}=\frac{\sqrt{2}}{2}$ (see Fig. \ref{2szqgb}). There are four peaks rogue wave on the double-periodic background when  $a=1$, $c=\frac{1}{2}$, $\beta_{3}=0.1$ and $\beta_{4}=\frac{\sqrt{2}}{2}$ (see Fig. \ref{4szqgb}).  Significantly, when $\beta_{3}$=$-\beta_{4}$, the rogue waves on the double-periodic background will convert to the rogue waves on the plane wave (see Fig. \ref{pmgb}). Due to the reverse-space-time reduction conditions of the eq.\eqref{ndnls}, the positions of the rogue wave solutions show the connections between reverse-space-time points $(x, t)$ and $(-x, -t)$.
We can verify this intuitively by observing the positions of two peaks rogue wave and four peaks rogue wave.

\begin{figure}[ht!]
\centering
\subfigure[First-order rogue wave on the double-periodic background with two peaks]{\label{2szqgb}
\begin{minipage}[b]{0.34\textwidth}
\includegraphics[width=4.5cm]{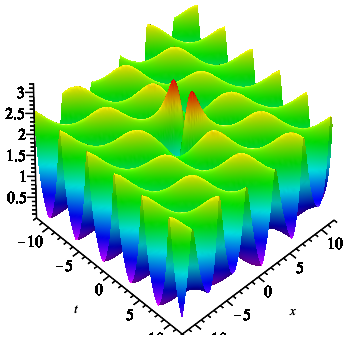}
\end{minipage}}
\subfigure[First-order rogue wave on the double-periodic background with four peaks]{\label{4szqgb}
\begin{minipage}[b]{0.3\textwidth}
\includegraphics[width=4.5cm]{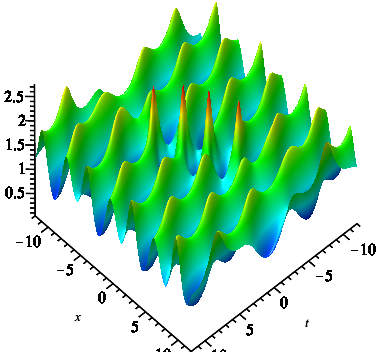}
\end{minipage}}
\subfigure[First-order rogue wave]{\label{pmgb}
\begin{minipage}[b]{0.3\textwidth}
\includegraphics[width=4.6cm]{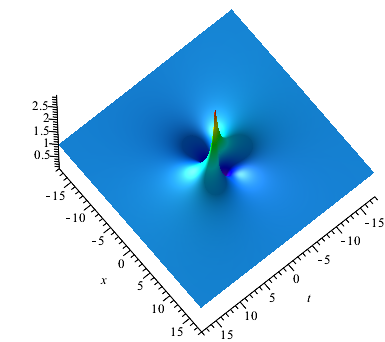}
\end{minipage}}
\caption{The dynamics of first-order rogue wave solution with: (a) $a=1$, $c=1$, $\beta_{3}=0.1$, $\beta_{4}=\frac{\sqrt{2}}{2}$; (b) $a=1$, $c=\frac{1}{2}$, $\beta_{3}=0.1$, $\beta_{4}=\frac{\sqrt{2}}{2}$; (c) $a=1$, $c=1$, $\beta_{3}=\frac{\sqrt{2}}{2}$, $\beta_{4}=-\frac{\sqrt{2}}{2}$.}
\label{q4gb}
\end{figure}

\par When $n=6$, we can obtain the second-order rogue waves on the double-periodic background by \eqref{gdtf}. Similar to the first-order rogue wave on double-periodic background, the selection of parameters also have effect both on the amplitude of the double-periodic background and rogue waves. The positions of the second-rogue waves also show the connections between reverse-space-time points $(x, t)$ and $(-x, -t)$. Compared with first-order rogue waves on the double-periodic background, the difference is that there are one peaks on the double-periodic background when we take $a=1$, $c=1$, $\beta_{5}=0.1$ and $\beta_{6}=\frac{\sqrt{2}}{2}$ (see Fig. \ref{2jgbdp2}). And there are two peaks on the double-periodic background when we take $a=1$, $c=\frac{1}{2}$, $\beta_{5}=0.1$ and $\beta_{6}=\frac{\sqrt{2}}{2}$ (see Fig. \ref{2j2}). See it visually in three dimensions, we can find that the energy centered on the rogue wave and gradually dissipates to a steady state. When $\beta_{5}$=$-\beta_{6}$, the second-order rogue waves on the double-periodic background will convert to the second-order rogue waves on the plane wave (see Fig. \ref{3j2}). In addition, compared with the first-order rogue waves, second-order rogue waves have higher amplitude.

\begin{figure}[ht!]
\centering
\subfigure[Second-order rogue wave on the double-periodic background with one peak]{\label{2jgbdp2}
\begin{minipage}[b]{0.3\textwidth}
\includegraphics[width=4.0cm]{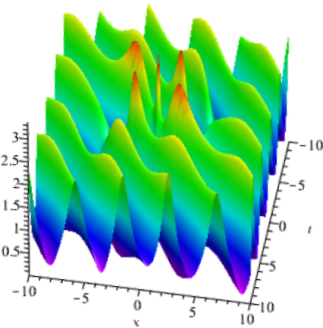}
\end{minipage}}
\subfigure[Second-order rogue wave on the double-periodic background with two peaks]{\label{2j2}
\begin{minipage}[b]{0.32\textwidth}
\includegraphics[width=4.2cm]{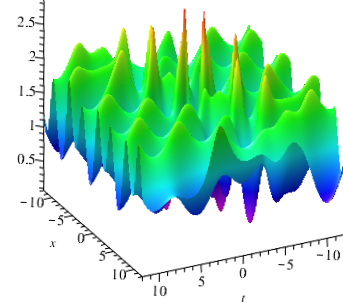}
\end{minipage}}
\subfigure[Second-order rogue wave]{\label{3j2}
\begin{minipage}[b]{0.3\textwidth}
\includegraphics[width=4.8cm]{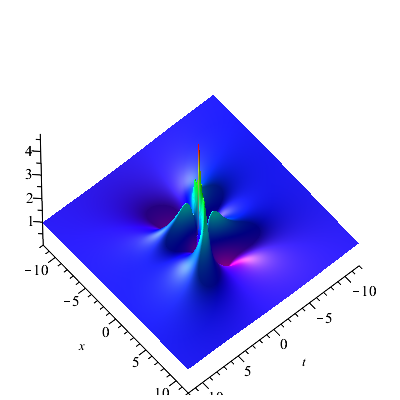}
\end{minipage}}
\caption{The dynamics of second-order rogue wave with:(a) $a=1$, $c=1$, $\beta_{5}=0.1$, $\beta_{6}=\frac{\sqrt{2}}{2}$; (b) $a=1$, $c=\frac{1}{2}$, $\beta_{3}=0.1$, $\beta_{4}=\frac{\sqrt{2}}{2}$; (c) $a=1$, $c=1$, $\beta_{3}=\frac{\sqrt{2}}{2}$, $\beta_{4}=-\frac{\sqrt{2}}{2}$.}
\label{q4gb}
\end{figure}
 \ \ \ \ \ \ \\

\section{Conclusion}
\ \ \ \
In our present investigation, we constructed the breathers and rogue waves on the double-periodic background for Eq. \eqref{ndnls}, which are first generated by plane wave seed solution. The general breathers, Ma breathers and Akmediev breathers on double-periodic background were generated by even-fold DT. Due to the influence of the double-periodic background, the image of Ma breathers solution looks like it disappears in the double-periodic background and the propagation direction of general breathers produces shift.

By using the even-order semi-degenerate DT, we derived the first-order and second-order rogue waves on the double-periodic background. Due to the reverse-space-time reduction conditions, the positions of rogue wave solutions show connections between reverse-space position and reverse-time points $(x, t)$ and $(-x, -t)$. For the first-order rogue waves on the double-periodic background, we find that there are two peaks or four peaks when we adjust the parameters. There are one peak or two peaks on the double-periodic background when adjusting parameters of second-order rogue waves. Second-order rogue waves have higher amplitude than first-order rogue waves. Significantly, the double-periodic background will convert to the plane wave background when $\beta_{n-1}=-\beta_{n}$.

We generated the general breathers, Ma breathers and Akhmediev breathers by odd-fold DT. There are some interesting phenomenons: the crest of the Ma breathers is cut and the phase shift occurs at the center of the general breathers with the perturbation of periodic background. The first-order and second-order rogue waves on the periodic background were derived by odd-order semi-degenerate DT formula, respectively. When $\beta_{n}>0$, the rogue wave patterns are located in the area where the periodic pattern reaches its amplitude. However, when $\beta_{n}<0$, the rogue wave patterns locate in the middle of two amplitude trajectories of the periodic pattern, which looks like that the rogue wave is generated by the interaction of two waves of the periodic pattern. The higher-order rogue waves have a high amplitude peak on the center distributed with some lower peaks and even numbers of caves.

Moreover, as we remarked in the introduction, rogue waves on the periodic and double-periodic background have some important uses and applications in many diverse areas of mathematics and physics. Therefore, the results which are presented in this article are also of great physical significance. For example, the rogue waves on the periodic and double-periodic background can be relevant to diagnostics of rogue waves on the ocean surface and understanding the formation of random waves due to modulation instability.


\begin{thebibliography}{}
\bibitem{ablowitz-prl-1973} Ablowitz, M.J., Kaup, D.J., Newell, A.C., Segur, H.:Nonlinear-evolution equations of physical signific, Phys. Rev. Lett. {\bf 31}, 125-122 (1973)
\bibitem{zakharov-springer-1991}  Zakharov, V.E.: What is integrability?. Springer, Berlin, (1991)
\bibitem{kaup-jmp-1978} Kaup, D.J., Newell, A.C.: An exact solution for a derivative nonlinear Schr\"{o}dinger equation. J. Math. Phys. {\bf 19}(4), 79801 (1978)
\bibitem{olver-springer-1990} Olver, PJ., Sattinger, D.H: Solitons in physics, mathematics, and nonlinear optics. Springer, New York (1990)
\bibitem{2-m-Abrarov-1992} Abrarov R.M., Christiansen P.L., Darmanyan S.A., Scott A.C., Soerensen M.P.: Soliton propagation in three coupled nonlinear Schr\"{o}dinger equations, Phys. Lett. A. {\bf 171}, 298-302 (1992)
\bibitem{yajima-jpsj-1995} Yajima, T: Derivative nonlinear Schr\"{o}dinger type equations with multipe components and their solutions. J. Phys. Soc. Jpn. {\bf 64}(6), 1901-1909 (1995)
\bibitem{lederer-pr-2008} Lederer, F., Stegeman, G. I., Christodoulides, D. N., Assanto, G.,  Segev, M., Silberberg, Y.: Discrete solitons in optics, Phys. Rep. {\bf 463} 1-126 (2008)
\bibitem{freak1996} Draper, L.: Freak ocean waves, Weather {\bf 21}, 2-4 (1966)
\bibitem{rogue2007nature} Solli, D.R., Ropers, C., Koonath, P., Jalali, B.: Rogue waves and rational solutions of the nonlinear Schr\"{o}dinger equation, Nature {\bf 450}, 1054 (2007)
\bibitem{rogue2009pra} Bludov, Y.V.,  Konotop, V.V., Akhmediev, N.: Matter rogue waves, Phys. Rev. A {\bf 80}, 033610 (2009)
\bibitem{frw2011} Yan, Z.Y.: Vector financial rogue waves, Phys. Lett. A {\bf 375} 4274-4279 (2011)
\bibitem{lcz2014} Li, C. Z.,  He, J. S., Porsezian, K.: Rogue waves of the Hirota and the Maxwell-Bloch equations. Phys. Rev. E {\bf 87} 012913 (2013).
\bibitem{breather2014} Dudley, J.M., Dias, F., Erkintalo, M., Genty, G.: Instabilities, breathers and rogue waves in optics, Nat. Photonics {\bf 8}, 755-764 (2014)
\bibitem{wanglei-annals-2015} Wang, L.,  Li, X.,  Qi, F. H., zhang, L. L.: Breather interactions and higher-order nonautonomous rogue waves for the inhomogeneous nonlinear Schr\"{o}dinger Maxwell-Bloch equations. Annals of Physics, {\bf 359} 97-114 (2015)
\bibitem{Liming-2018} Ling, L. M.,  Zhao, L. C., Yang, Z. Y., Guo, B. l.: Generation mechanisms of fundamental rogue wave spatial-temporal structure, Phys. Rev. E {\bf 96}, 022211 (2017)
\bibitem{jxw2020} Jin, X.W.,  Lin, J.: Rogue wave, interaction solutions to the KMM system, J. Magn. Magn. Mater. {\bf 502}, 166590 (2020)
\bibitem{JPAXSW} Xu, S.W., He, J.S., Wang, L.H.: The Darboux transformation of the derivative nonlinear Schr\"{o}dinger equation, J. Phys. A-Math. Theor. {\bf 44}, 6629-6636 (2011)
\bibitem{zys2014} Zhang, Y.S.,  Guo, L.J., Xu, S.W., Wu, Z.W., He, J.S.: The hierarchy of higher order solutions of the derivative nonlinear Schr\"{o}dinger equation, Commun. Nonlinear Sci. {\bf 19}, 1706-1722 (2014)
\bibitem{xt2018nd} Xu, T., Chen, Y.: Mixed interactions of localized waves in the three-component coupled derivative nonlinear Schr\"{o}dinger equations, Nonlinear Dyn. {\bf 92}, 2133-2142  (2018)
\bibitem{2-m-Kundu-2010} Kundu, A.: Two-fold integrable hierarchy of nonholonomic deformation of the derivative nonlinear Schr\"{o}dinger and the Lenells-Fokas equation, J. Math. Phys. {\bf 51}, 022901 (2010)
 \bibitem{plasma1} Mj$\phi$lhus, E.: On the modulational instability of hydromagnetic waves parallel to the magnetic field, J. Plasma Phys. {\bf 16}, 321-334 (1976)
     \bibitem{magnetic} Lakhina, G.S.,  Sharma, A.S.,  Buchner, J.: International workshops on nonlinear waves and chaos in space plasmas-preface, Nonlinear Proc. Geoph. {\bf 11}(2), 181-181 (2004)
\bibitem{plasma2} Ruderman, M.S.: DNLS equation for large-amplitude solitons propagating in an arbitrary direction in a high-$\beta$ hall plasma, J. Plasma Phys. {\bf 67}, 271-276 (2002)
\bibitem{plasma3} Shan, S.A., El-Tantawy, S.A.: The impact of positrons beam on the propagation of super freak waves in electron-positron-ion plasmas, Phys. Plasmas {\bf 23}(7), 072112 (2016)
\bibitem{7} Tzoar, N., Jain, M.: Self-phase modulation in long-geometry optical waveguide, Phys. Rev. A. {\bf 23}, 1266-1270 (1981)
\bibitem{8} Anderson, D., Lisak, M.: Nonlinear asymmetric self-phase modulation and self-steepening of pulses in long optical waveguides, Phys. Rev. A {\bf 27}, 1393-1398 (1983)
\bibitem{9} Govind, P.A.: Nonlinear fibers optics 3rd edn, New York: Adademic (2001)
\bibitem{kn1994} Zeng, Y.: New factorization of the Kaup-Newell hierarchy, Physica D. {\bf 73}, 171-188 (1994)
\bibitem{kn1999-zz} Zhou, Z.X.: Parameters of darboux transformation for reduced akns, kaup-newell and pcf systems, Chinese Ann. Math. B {\bf 20}, 195-204 (1999)
\bibitem{CNSCZZX} Zhou, Z.X.: Darboux transformations and global solutions for a nonlocal derivative nonlinear Schr\"{o}dinger equation, Commun. Nonlinear Sci. {\bf 62}, 480-488 (2016)
\bibitem{29i} Matveev, V.B., Salle, M.A.: Darboux transformations and solitons, Springer, Berlin-Heidelberg (1991)
\bibitem{VBDB} Li, Y.S.: Soliton and integrable system, Shanghai Sci.-Tech. Edu., Publishing House, Shanghai (1991)
\bibitem{CHDB} Gu, C.H.: Darboux transformation in soliton theory and its geometric applications, Shanghai Sci.-Tech. Edu., Publishing House, Shanghai (2005)
    \bibitem{2-m-GuChaohao-2005} Gu, C.H., Hu, H.S., Zhou, Z.X.: Darboux Transformations in Integrable Systems: Theory and Their Applications (Berlin: Springer) 2005.
\bibitem{xtzs2017} Xu, T., Li, H.j., Zhang, H.j., Li, M., Lan, S.: Darboux transformation and analytic solutions of the discrete PT-symmetric nonlocal nonlinear Schr\"{o}dinger equation,  Appl. Math. Lett. {\bf 63}, 88-94 (2017)
\bibitem{wmmnd} Wang, M.M.,  Chen Y.: Dynamic behaviors of mixed localized solutions for the three-component coupled Fokas-Lenells system. Nonlinear Dyn. {\bf 98}(3), 1781-1794 (2019)
\bibitem{NDSY} Shi, Y., Shen, S.F., Zhao, S.L.: Solutions and connections of nonlocal derivative nonlinear Schr\"{o}dinger equations, Nonlinear Dyn. \textbf{95}, 1257-1267 (2019)
\bibitem{MPLBMDX} Meng, D.X.,  Li, K.Z.: Darboux transformation of the second-type nonlocal derivative nonlinear Schr\"{o}dinger equation, Mod. Phys. Lett. B {\bf 33}(10), 1950123 (2019)
\bibitem{ab-prl-2013} Ablowitz, M.J., Musslimani, Z.H.: Integrable nonlocal nonlinear Schr\"{o}dinger equations, Phys. Rev. Lett. {\bf 110}, 064105 (2013)
    \bibitem{yjk-pla-2019} Yang, J.K.: General N-solitons and their dynamics in several nonlocal nonlinear Schr\"{o}dinger equations, Phys. Lett. A {\bf 383}(4), 328-337 (2019)
\bibitem{ab-studies-2017} Ablowitz, M.J., Musslimani, Z.H.: Integrable nonlocal nonlinear equations, Stud. Appl. Math. {\bf 139}(1), 7-59, (2017)
\bibitem{gad-pra-2016} Gadzhimuradov, T. A., Agalarov, A. M.: Towards a gauge-equivalent magnetic structure of the nonlocal nonlinear Schr\"{o}dinger equation, Phys. Rev. A {\bf 93} 062124 (2016).
\bibitem{yjk-pre-2018} Yang, J. K.: Physically significant nonlocal nonlinear Schr\"{o}dinger equation and its soliton solutions, Phys.Rev. E {\bf 98}, 042202,  (2018)
    \bibitem{pjc-nd-2021}  Pu, J. C.,  Li, J.,  Chen, Y.: Solving localized wave solutions of the derivative nonlinear Schr\"{o}dinger equation using an improved PINN method. Nonlinear Dyn. {\bf 105}, 1723-1739 (2021).
\bibitem{zgq2017} Zhang, G.Q., Yan, Z.Y, Chen, Y.: Novel higher-order rational solitons and dynamics of the defocusing integrable nonlocal nonlinear Schr\"{o}dinger equation via the determinants.  Appl. Math. Lett. {\bf 69}, 113-120 (2017)
\bibitem{ablowitz-tmp-2018}Ablowitz, M.J., Feng, B. F., Luo, X. D.,  Musslimani, Z. H.: Inverse scattering transform for the nonlocal reverse space-time nonlinear Schr\"{o}dinger equation, Theor. Math. Phys.  {\bf 196}, 1241-1267 (2018)
\bibitem{kaa-2014} Kedziora, D.J.,  Ankiewicz, A.,  Akhmediev, N.: Rogue waves and solitons on a cnoidal background, Eur. Phys. J-Spec. Top. {\bf 223}(1), 43-62 (2014)
\bibitem{cjb2018non} Chen, J.B.,  Pelinovsky, D.E.: Rogue periodic waves of the modified KdV equation, Nonlinearity {\bf 31}, 1955-1980 (2018)
\bibitem{rjg2018} Rao, J.G.,  Zhang, Y.S., Fokas, A.S., He, J.S.: Rogue waves of the nonlocal Davey-Stewartson I equation, Nonlinearity {\bf 31}(9), 4090-4107 (2018)
\bibitem{cjb-2019} Chen, J.B.,  Pelinovsky, D.E., White, R.E.: Rogue waves on the double-periodic background in the focusing nonlinear Schr\"{o}dinger equation, Phys. Rev. E {\bf 100}, 052219 (2019)
\bibitem{xbgxg2020} Xue, B., Shen, J.,  Geng, X.G.: Breathers and breather-rogue waves on a periodic background for the derivative nonlinear Schr\"{o}dinger equation, Phys. Scripta {\bf 95}(5), 055216 (2020)
\bibitem{lb2020} Liu, Y.,  Li, B.: Dynamics of solitons and breathers on a periodic waves background in the nonlocal Mel'nikov equation, Nonlinear Dyn. {\bf 100}(4), 3717-3731 (2020)
\bibitem{zhq2021} Zhang, H.Q., Chen, F.,  Pei, Z.J.: Rogue waves of the fifth-order Ito equation on the general periodic travelling wave solutions background, Nonlinear Dyn. {\bf 103}, 1023-1033 (2021)
\bibitem{sms-2021} Sinthuja, N.,  Manikandan, K., Senthilvelan, M.: Rogue waves on the double-periodic background in Hirota equation, Eur. Phys. J. Plus. {\bf 136}(3), 1-12 (2021)
\bibitem{hx2016epjp} Huang, X., Ling, L.M.: Soliton solutions for the nonlocal nonlinear Schr\"{o}dinger equation. Eur. Phys. J. Plus {\bf 131}, 148 (2016)
\bibitem{hjs2013nmp} He, J.S.,  Tao, Y.S.,  Porsezian, K., Fokas, A.: Rogue wave management in an inhomogeneous Nonlinear Fibre with higher order effects, J. Nonlinear Math. Phys. {\bf 20}, 407-419 (2013)
\bibitem{hjs2014pla} He, J.S., Charalampidis, E.G.,  Kevrekidis, P.G., Frantzeskakis, D.J.: Rogue waves in nonlinear Schr\"{o}dinger models with variable coefficients: Application to Bose-Einstein condensates, Phys. Lett. A {\bf 378}(56), 577-583 (2014)
\bibitem{zlc2017cnsns} Zhao, L.C., Ling, L.M., Qi, J.W., Yang, Z.Y.,  Yang, W.L.: Dynamics of rogue wave excitation pattern on stripe phase backgrounds in a two-component Bose-Einstein condensate, Commun. Nonlinear Sci. Numer. Simulat, {\bf 49}, 39-47 (2017)
\bibitem{lwhjs2018} Liu, W., Zhang, Y.S., He, J.S.: Rogue wave on a periodic background for Kaup-Newell equation, Rom. Rep. Phys. {\bf 70}, 106 (2018)
 \bibitem{dccgyt2019} Ding, C.C.,  Gao, Y.T.,  Li, L. Q.: Breathers and rogue waves on the periodic background for the Gerdjikov-Ivanov equation for the Alfv\'{e}n waves in an astrophysical plasma, Chaos Soliton. Fract. {\bf 120}, 259-265 (2019)
      \bibitem{sr2018pre} Randoux, S., Suret, P.,  Chabchoub, A.,  Kibler, B.,  El, G.: Nonlinear spectral analysis of Peregrine solitons observed in optics and in hydrodynamic experiments, Phys. Rev. E {\bf 98}, 022219 (2018).
    \bibitem{ac-wavemotion-2017} Calini, A., Schober, C. M., Characterizing JONSWAP rogue waves and their statistics via inverse spectral data, Wave Motion, {\bf 71}, 5 (2017).
\bibitem{FEPA} Fan, E.G.: A Liouville integrable Hamiltonian system associated with a generalized Kaup-Newell spectral problem, Physica A. {\bf 301}, 105-113 (2001)
\bibitem{MWJPA} Ma, W.X., Zhou, R.: A coupled AKNS-Kaup-Newell soliton hierarchy, J. Math. Phys. {\bf 40}(9), 4419-4428 (1999)
    \end{thebibliography}
\end{document}